%% file: rbcpd-acm.tex
\def\full{0}

\ifodd\full%
\newcommand{\ifFull}[1]{#1}
\newcommand{\ifFullElse}[2]{#1}
\else
\newcommand{\ifFull}[1]{}
\newcommand{\ifFullElse}[2]{#2}
\fi

\documentclass[sigconf]{acmart}

\AtBeginDocument{%
  \providecommand\BibTeX{{%
    \normalfont B\kern-0.5em{\scshape i\kern-0.25em b}\kern-0.8em\TeX}}}

\setcopyright{acmcopyright}
\copyrightyear{2018}
\acmYear{2018}
\acmDOI{10.1145/1122445.1122456}

\acmConference[Woodstock '18]{Woodstock '18: ACM Symposium on Neural
  Gaze Detection}{June 03--05, 2018}{Woodstock, NY}
\acmBooktitle{Woodstock '18: ACM Symposium on Neural Gaze Detection,
  June 03--05, 2018, Woodstock, NY}
\acmPrice{15.00}
\acmISBN{978-1-4503-XXXX-X/18/06}

\begin{document}

\title{Decentralizing Information Technology: The Advent of Resource Based Systems }

\author{Aggelos Kiayias}
\email{aggelos.kiayias@ed.ac.uk}
\orcid{0000-0003-2451-1430}
\affiliation{%
  \institution{University of Edinburgh \& IOG}
  \streetaddress{10 Crichton St.}
  \city{Edinburgh}
  \country{UK}
  \postcode{EH8 9AB}
}

\renewcommand{\shortauthors}{Kiayias}

\begin{abstract}
\input{abstract.tex}
\end{abstract}



\maketitle

\def\mc{\mathcal}
\newcommand{\ignore}[1]{}

\section{Introduction}

\input{introduction.tex}

\section{Fundamental Characteristics of Resource Based Systems}

\input{characteristics.tex}

\section{Resource-Based Participation}

\input{pox.tex}

\section{Tokenomics}
\label{sec:tokenomics}

\input{tokenomics.tex}

\section{Decentralized Service Provision}

\input{service.tex}

\section{Rewards Sharing}
\label{sec:rewards}

\input{rewards.tex}

\section{A high-level blueprint for a stake-based system}
\input{blueprint.tex}

\section{Conclusing remarks}
\label{sec:remarks}

\input{conclusion.tex}

\section{Acknowledgements}

\input{acknowledgements.tex}

\bibliographystyle{plain}
\bibliography{consensus,crypto}

\end{document}
\endinput

%% file: abstract.tex
The growth of the Bitcoin network during the first decade of its operation to a global scale system is a singular event in the deployment of Information Technology systems. Can this approach serve as a wider paradigm for Information Technology services beyond the use case of digital currencies? We investigate this question by introducing the concept of resource based systems and their four fundamental characteristics: (i) resource-based operation, (ii) tokenomics, (iii) decentralized service provision, and (iv) rewards sharing. We explore these characteristics, identify design goals and challenges and investigate some crucial game theoretic aspects of reward sharing that can be decisive for their effective operation. 

%% file: introduction.tex
A paradigm shift took place during the last decade in the way the consensus problem is looked at in Computer Science. Three decades after the seminal work of Lamport, Shostak and Pease \cite{DBLP:journals/jacm/PeaseSL80}, Satoshi Nakamoto with the Bitcoin blockchain protocol \cite{Nakamoto2008} put forth a novel way to solve the consensus problem. Traditionally, Byzantine consensus was considered to be the problem of reaching agreement between a set of processors, some of which may arbitrarily deviate from the protocol and try to confuse the ones who follow the protocol. Over time, significant research was invested into establishing the exact bounds in the {\em number} of deviating parties as well as the intrinsic complexity bounds of the problem in terms of round and message complexity. 

The approach did not address the question who assigns identities to the processors or sets up the network that connects them, or how the processors agree about the identities of all those present in the particular protocol instance.  These were 
tasks left to a system setup phase that, for all purposes, seemed sufficient to be a centralized  operation carried out by a trusted party. The success of the Internet however and the development of peer-to-peer networks in the early 2000's 
set the stage for challenging this assumption. At the same time, Sybil 
attacks \cite{DBLP:conf/iptps/Douceur02} posed a significant obstacle
to apply known consensus protocol techniques to this new setting.   

Given the above landscape, Nakamoto's solution is unexpected. 
The blockchain protocol design circumvents entirely the issue of identity and provides a solution for consensus that ``takes advantage of 
information being easy to spread but hard to stifle'' \cite{Nakamoto2009}.
In the Bitcoin blockchain, it is the computational power that different participants contribute to the protocol execution that facilitates convergence to a unique view. And as long as the deviating parties are in the minority of computational power, Nakamoto's blockchain protocol can be shown to converge to an unambiguous view that incorporates new inputs in a timely manner, as proven 
formally in      \cite{DBLP:conf/eurocrypt/GarayKL15} and subsequently refined in~\cite{DBLP:conf/eurocrypt/PassSS17,DBLP:conf/crypto/GarayKL17}.

The bitcoin blockchain however is much more than just a consensus protocol; it provides a service --- transferring digital currency between principals --- and also provides incentives  to those who engage in service provision in the form of ``mining'' the digital currency following a regular schedule. In this way, agreeing to  a joint view is not the primary objective but rather a precondition for the service to materialize. Participants need to agree on the ledger of digital asset ownership. 

Becoming a system maintainer in Bitcoin does not require anything else other than possessing the software and necessary hardware to run the protocol. 
Remarkably, there is no need to be approved by other participants as long as the underlying peer to peer network allows diffusing messages across all peers without censorship. 

Given the successful growth of the Bitcoin network at a worldwide level, a fundamental question arises: does its architecture suggest a novel way of deploying information technology (IT) services at a global scale? So far, in IT, we have witnessed two major ways of scaling systems to such level. In the ``centralized''  approach, we have organizations such as Google, Facebook and Amazon that offer world-wide operations with high quality of service. The downsides of the centralized approach is ---naturally--- centralization itself and the fact that long term common good may not necessarily align properly with company shareholder interest. In such settings,  regulatory arbitrage may deprive the public  any leverage against the service operator. A second approach is the ``federated''  one. In this case, we have the coordination of multiple organizations or entities with a diverse set of interests to offer the service in cooperation. Examples of such federated organization have been very successful and far reaching as they include the Internet itself. Nevertheless, for federated organization to  scale,  significant efforts should be invested to make sure the individual systems interoperate and the incentives of their operators remain aligned. 

Viewed in this light, the decentralized approach offered by Naka\-moto's Bitcoin system provides an alternative pathway to a globally scalable system. 
In this paper, we abstract Nakamoto's novel approach to system deployment under a general viewpoint that we term ``resource-based systems.'' 
Preconditioned on the existence of an open network, 
a resource based system capitalizes on a particular resource 
that is somehow dispersed across a wide population of {\em resource holders} and  bootstraps itself to a sufficient level of quality of service and security 
out of the self-interest of resource holders who  engage with each other via the system's underlying protocol. Depending on the design, the resulting system may scale more slowly and be less performant than a centralized system, but it offers  security and resilience characteristics that can be attractive for a number of applications especially for global scale IT. 

In the next section we introduce the four fundamental characteristics of resource based systems and then we expand on each one in a separate section describing the challenges that resource based system designers must overcome in order to 
make them successful. We also identify an end-point in the effective operation of a resource based system that arises in the form of a centralized equilibrium. We discuss the ramifications of this result and conclude with some examples of resource based systems and directions for future research. 

%% file: characteristics.tex
Consider a service described as a program $\mc{F}$ that  captures all 
the operations that users wish to perform. 
%
%
A resource-based realization of $\mc{F}$ 
is a system  that exhibits the following four fundamental characteristics, cf. Figure~\ref{fig:4c}

\begin{figure}[h]
  \includegraphics[width=\linewidth]{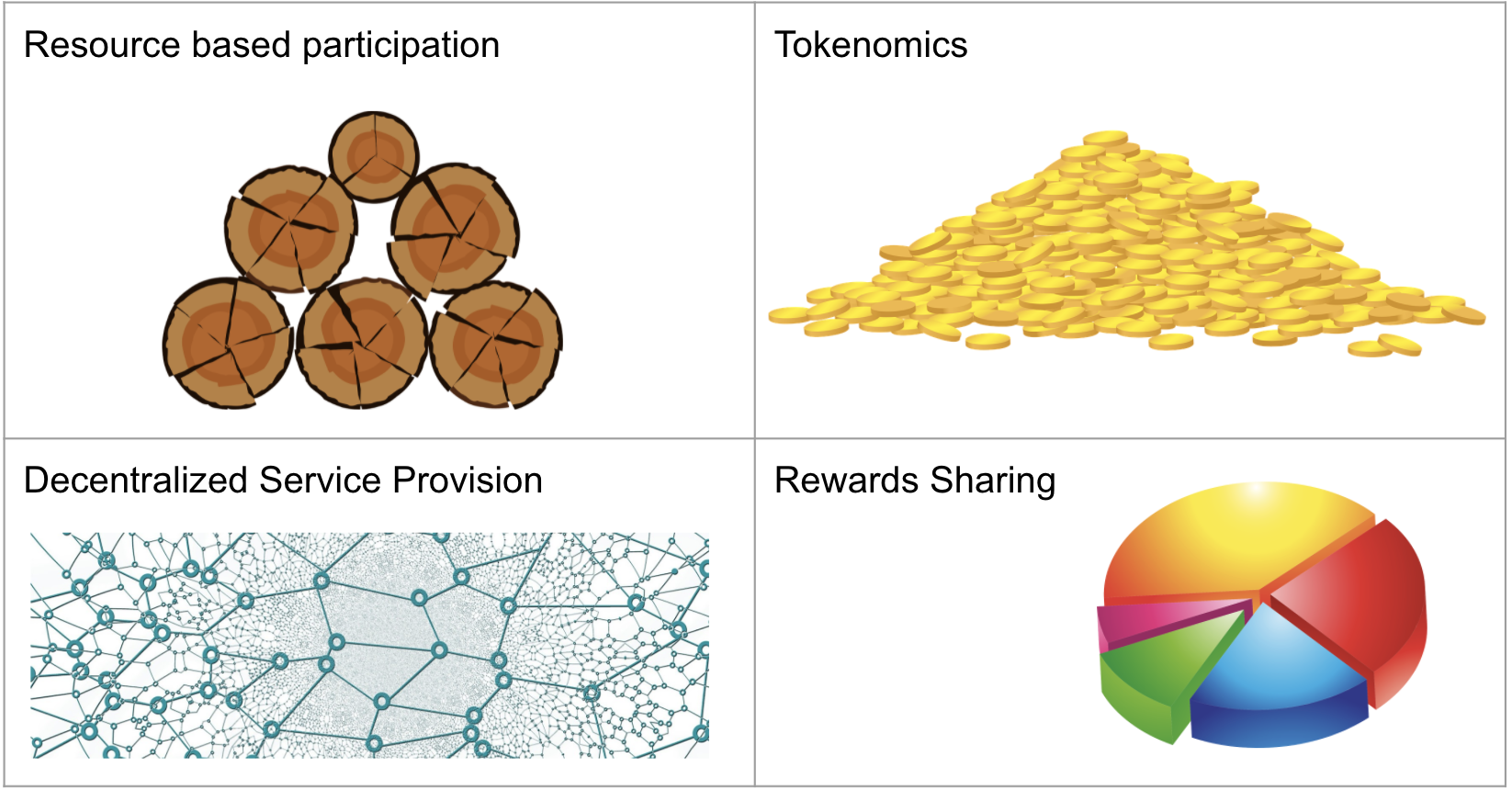}
  \caption{\label{fig:4c} The four characteristic components of resource-based systems}
\end{figure}

\begin{itemize}
\item {\em Resource-based participation.} There exists a fungible 
resource that can be acquired by anyone interested in doing so, possibly at a cost. 
Entities in possession of units of the resource can exercise it 
to participate in the maintenance of the service, possibly incurring 
further costs. 

\item {\em Tokenomics.}
The system issues a digital asset or currency that 
is used to tokenize the collective efforts of the maintainers
and reward them. 
Such digital ``coins''  are  maintained in cryptographic 
wallets and should be argued to be of sufficient utility to make system 
maintenance an attractive endeavor as a joint effort. 

\item {\em Decentralized service provision.} 
A user interacts with the service by submitting a transaction
which is openly circulated in the network of maintainers
provided it is well formed. Such well formedness 
may require the commitment of  a sufficient amount of
digital currency  or other user expenditure to prevent spamming. 
The maintainers collectively take the required actions
of $\mc{F}$ needed for the submitted transactions 
in a consistent and expedient fashion while 
the system records their efforts in a robust manner. 

\item {\em Rewards Sharing.} 
The digital assets that the system makes available
to maintainers are distributed to the active
maintainers in a regular and fair manner so that
the system's safety and liveness properties emanate
from their incentive driven participation. Any property violation 
should be a deviation from an equilibrium state that incurs costs 
to the perpetrators, hence ensuring the stable operation of the system. 
\end{itemize}

Based on the above, implementing a given system functionality 
 $\mc{F}$ by a resource based protocol requires the 
 design of a protocol with suitable cryptographic 
 and game theoretic characteristics. 
In the next four sections we delve into each characteristic in
some more detail. 

%% file: pox.tex
In classic distributed systems, system maintenance is offered by nodes that
are authorized to do so either by common agreement  (e.g., via a list of public-keys that identify them)
or by the network connections that are assumed to exist between them 
(cf.  \cite{sok-gk20} for an overview). 
Such configurations are commonly referred to in cryptographic modeling 
as setup assumptions.

Contrary to this, in the decentralized setting of 
 resource-based  protocols, participation 
to contribute to the protocol execution is attained
via demonstrating the possession of a certain resource. 
This comes in the form of a {\em proof of resource}, commonly
referred to as ``proof-of-X'' (PoX) in Bitcoin nomenclature,
where X signifies the particular resource in question. 

It is worth noting that this approach generalizes both the classic
distributed setting, since the resource in question could be the possession
of one of the authorized identities, as well as the centralized setting
--- in a trivial manner.

The two most widely cited such schemes are proof-of-work (PoW) 
and proof-of-stake (PoS). The case of PoW is exemplified in the 
Bitcoin blockchain protocol~\cite{Nakamoto2008} and is essentially a proof of 
possession of computational power. Given the characteristics  
of the PoW algorithm, a specific  
logic or architecture may be more advantageous 
and as a result, maintainers  may benefit from special purpose 
implementations. In such case, the PoW algorithm will  not be a proof
of general computational power, but rather a proof of 
ability to execute the particular algorithm utilized in the PoW scheme. 
This issue has brought forth significant criticism against   
the implementation of PoW utilized in the Bitcoin protocol (the  hashcash algorithm \cite{hashcash} instantiated by the hash function SHA256) and subsequently a number of other PoW algorithms 
were developed and deployed in alternative blockchain protocols
(these include scrypt, see \cite{DBLP:journals/rfc/rfc7914}, and 
ethash, see \cite{wood2014ethereum}, which motivated research in memory-hardness
with algorithms such as argon2 \cite{DBLP:journals/rfc/rfc9106}). 

Independently of the properties of the implementation, 
a common characteristic of PoW is that running the algorithm 
requires energy expenditure (needed to power the hardware executing
the algorithm logic). This aspect, combined with the fact that 
the source of energy cannot be discriminated, lead to concerns about
the use of non-renewable sources in the Bitcoin blockchain. 

Contrary to PoW, a PoS scheme proves possession of a virtual resource
(e.g., the possession of a certain amount of digital currency). 
A significant distinction in this class of algorithms is that issuing a PoS 
has cost independent of the amount of ``stake'' in  possession 
of the prover, while PoW typically incurs a linear cost in 
terms of computational power. Examples of PoS schemes are Ouroboros~\cite{DBLP:conf/crypto/KiayiasRDO17} and Algorand \cite{DBLP:conf/sosp/GiladHMVZ17}. 

Beyond stake and work, other types of resources, both virtual and 
physical, have been proposed
and utilized. These include ``proof of space'', whereby the prover
demonstrates possession of storage capacity, cf.\cite{DBLP:conf/crypto/DziembowskiFKP15},  and
``proof of elapsed time'', supported by Intel SGX cf. \cite{sawtooth}, whereby the 
prover demonstrates that certain wait time has elapsed, just to name
two examples. 

An important property of PoX's in the way  they are integrated
within the underlying system
is the fact that the freshness of the proof needs
 to be ensured. 
This is invariant of the specific resource used. 
In particular, it should be impossible
to ``replay'' an old proof that refers to resources possessed
at a previous point in time.  This point is crucial since 
resources are transferable between participants and hence
any proof should reflect the present state of   
resource allocation. 

A second property that is also essential is that the verification
of a PoX should be performed with low complexity, ideally 
independent of the level of resources involved in the generation
of the proof. The reasoning behind this requirement is that 
verification is something that needs to be performed network
wide, possibly even by entities that do not possess any units of the resource
used in PoX, as such entities may still need to verify the state of the system. 

%% file: tokenomics.tex
The key concept behind resource-based system tokenomics is 
the {\em tokenization} 
of the efforts of the system maintainers in the form  of a {\em digital asset}
that subsequently can be utilized in some way by the maintainers
to influence the utility they extract from the system. 
The essential objective is that ---collectively--- maintainers' utility remains
positive and hence maintenance costs can be  covered and the system is viable.
The necessary economics argument needed here gives rise to the term ``tokenomics''  as a
portmanteau word derived from ``token''  and ``economics.''

The approach  to achieve the tokenomics objective suggested by Nakamoto's design and built upon and extended in numerous follow up blockchain projects is {\em market based}. The underlying digital asset of the system becomes a native digital currency that is required for accessing the service. The system also facilitates the exchange of the digital currency between parties and hence  a market is created for the IT service. 
Moreover, its availability for public trading allows speculators to estimate the value of the service in the near term and far future. 

At system launch, it is possible to have  a pre-distribution of digital coins. For instance, digital coins can be ``airdropped'' to token holders of a pre-existing digital currency. In other cases, digital coins can be made available to investors in a ``pre-sale'' stage whereby software developers of the platform may use to fund the development of the software pre-launch, cf. \cite{icos} for an exposition of some of the relevant economics considerations in this setting. 

A key characteristic of the digital currency is that it should be easily transferable between parties. The coin can be listed on ``cryptocurrency''  exchanges and hence its value can be determined vis-\`a-vis other currencies (or commodities or other instruments) that potential system users may possess. 
Users should be able to keep the digital coin on a ``wallet'', an application  running on a user's device.  This means that a  user should apply a cryptographic key in order to exercise ownership of the digital coins and issue a transaction. While users of the system have the option to perform the necessary key management themselves, they can also opt to delegate the custody of their digital assets to third parties. 

After system launch, additional coins can become available and distributed following a certain  ruleset to the system maintainers. Such ruleset is 
public knowledge and algorithmically enforced during system maintenance. 
Depending on the system, the rate of new coin availability can be constant
per unit of time (as e.g., in Ethereum originally
), or
follow some function per unit of time (as e.g., in the case of Bitcoin, which has a finite total supply and hence relies on a geometric series to distribute coins to maintainers). 
In some cases, the  number of new coins that become added to the circulating supply depend on the behavior of the maintainers, e.g., in the case of Cardano, higher coin ``pledges''  by the participants increase the rate that coins become available. While the ruleset is algorithmic, enforced in the system ``ledger rules'', it can be changed by modifying the software that supports the system, assuming there is wide consensus between the system maintainers to adopt the update (see \cite{DBLP:conf/esorics/CiampiKKZ20} for a formal treatment of updates in blockchains). An example of such  update was for instance the ``London update''  of Ethereum which made the total supply of coins a variable function that depends on the transactions processed by the system. 

In some cases, an amount of coins is reserved to a development fund or 
``treasury''  that the system can subsequently utilize  to fund further
research and development. This is exemplified in systems like Decred\footnote{See \url{https://docs.decred.org/research/overview/}.}
and Cardano's Project catalyst\footnote{See \url{https://projectcatalyst.org}.} and also explored from first
principles, see e.g., \cite{DBLP:conf/ndss/ZhangOB19}. Such treasury
can also receive funding in perpetuity by taxing the system users. 

The ability to trade the system's coin enables the assessment of the system's
potential value given its public tokenomics characteristics and its 
utility.  A plunge in value of  the system's coin suggest a loss of faith
in the system's utility   and can lead to system 
instability due to weak participation in system maintenance or 
the total abandonment of maintenance and the inevitable ``death'' of the system. 
Examples of system instability are for instance the 51\% attacks
observed against a number of systems (e.g., Ethereum classic in \cite{ethc})
while there are thousands of cryptocurrency projects that have been
abandoned.\footnote{See for instance, 
\url{https://www.coinopsy.com/dead-coins/}. }
Trading also enables reinvesting rewards towards the acquisition of 
additional resources for protocol participation leading to interesting
compounding effects, see e.g. \cite{DBLP:conf/fc/FantiKORVW19,DBLP:conf/tokenomics/KarakostasKNZ19} for some interesting investigations in this direction. 

While speculation can drive the price of the system's digital currency up, 
the real value of the token lies with the underlying functionality
of the system. Ultimately, users should wish to obtain the token in order
to transact with the system and engage in the functionality it offers. 
Given this, speculators may choose to acquire the token at an early time 
prior to the system's utility becoming fully realized in the hope of 
selling the token for profit at a future time. 

To conclude, the key requirement for the above market-based tokenomics approach to achieve the essential objective identified in the beginning of this section  is the following:  the demand curve
for the system token, as a function of time, when projected over the supply as determined by the digital currency schedule,  should produce a  price for the system token that  offsets the collective cost of maintenance  for 
a sufficient level of  quality of service. 
This ensures that system maintenance, at least in the way it is encoded by the system software,  is an attractive endeavor to engage in.
As a final note, it is worth noting that the market-based approach, even though currently dominant in deployed resource based systems, may not be necessary for successful system tokenomics. In theory other mechanisms could also work, e.g., a reputation-based approach. As long the system is able to influence the  utility of system mainainers to a sufficient degree so that quality of service is maintained, the tokenomics objective would have been met.

%% file: service.tex
Now we come to the key question how to offer the prescribed
service functionality $\mc{F}$ while users have no specific point of contact
to reach. Users package their desired input to a transaction and release that transaction to 
the open network. The software running in the system maintainer side 
should receive such transaction and in the right circumstances 
propagate it to the network. 
The system maintainers should
act on this input in a timely manner and collectively the system
should reach a state where the desired output is produced. 

There are two important properties that the decentralized
implementation of the service should provide: safety and liveness. 

In terms of safety, the condition that is sought is that the system
should exhibit consistency in the way it processes requests. 
In other words, despite being realized by a fluctuating population
of system maintainers, the resulting effect of applying
a certain valid input should never be incongruous  to the 
effect that the same input would have if it was applied to $\mc{F}$. 
A safety violation for instance, would be that a user submits
two mutually exclusive  inputs, i.e., inputs that cannot be both 
applied to the state of $\mc{F}$ and subsequently some users
observe the first input as being actioned upon by the system 
while others observe similarly the second input. 

The liveness property refers to the ability of the 
system to react in a timely manner to users' input.  
Liveness, may be impacted both by congestion or even 
denial of service (DoS) attacks,  where the system's capacity
gets depleted as well as by censorship attacks where 
system maintainers choose to ignore the user's input. 
The expected responsiveness of the system 
may be affected by demand, but ideally, the system should
have the capability to scale up with increasing demand
so that quality of service is maintained. 

Given the decentralized nature of the implementation, it would
be impossible to argue the above properties in all circumstances. 
Instead, what is sought is to argue  that under reasonable
resource restrictions the properties hold in the Byzantine
sense -- i.e., an adversary cannot violate them unless it controls
a  significant amount of resources. It is worth pointing out  that 
while this is necessary,  it is not sufficient;  we would also want that given a reasonable 
modeling of the participants' utility functions, the desired
system behavior 
is an equilibrium or even a dominant strategy, given a plausible class of 
demand curves for service provision. Such strategic considerations  however
will have to wait for the next Section on reward sharing.

While system maintainers utilize their resources to  support
service provision, it should be feasible for clients of the system's
functionality $\mc{F}$ to engage in a ``light weight'' fashion with 
the system, i.e., while spending the minimum effort possible both
in terms of communication as well as computational complexity 
and resource expenditures. 

Mitigating DoS attacks can be a crucial consideration
in resource-based systems given their open and distributed nature. 
For instance, it should be hard for malicious actor to generate
a significant load of ``spam'' transactions and saturate the system's
capacity. Collecting fees to process transactions in a native digital 
currency of the resource based system  is a standard 
way that can help mitigate such DoS attacks while it also helps
generating revenue for the maintainers that is proportional to 
the transactions processed. 

A final important component of the system implementation 
is the ability of the system to collectively record  at regular intervals relevant  {\em performance metrics}  for  the system maintainers that engage with it.
While not needed for providing the service to the clients, metrics are important
so that the system records  the efforts of maintainers so they
can be rewarded appropriately. The performance metrics operation should 
be {\em robust} in the sense that, ideally,
the metrics are resilient in the Byzantine 
sense: a set of maintainers, perhaps appropriately restricted, 
should not be able to manipulate the recorded performance
of a targeted system maintainer. 
The Bitcoin blockchain uses a non-robust performance metric
(the number of blocks produced by a system maintainer) which has given
rise to attacks (cf. selfish-mining \cite{eyalsirer2014}). Other blockchain
protocols in the PoW and PoS setting developed robust metrics, 
see \cite{DBLP:conf/podc/PassS17,DBLP:conf/crypto/KiayiasRDO17}, enabling 
better game theoretic argumentation --- e.g., proving the protocol an equilibrium in \cite{DBLP:conf/crypto/KiayiasRDO17}. 

%% file: rewards.tex


In this section we come to the topic of ``rewards sharing'' 
that focuses on how individual system maintainers are being compensated and the strategic considerations that arise from that. 
As mentioned in section~\ref{sec:tokenomics} the system may make digital coins available to the maintainers following a specific schedule. Additionally, maintainers may claim transaction fees that are provided by the users 
who engage with the system. 

The operation of rewards sharing can be ``action based'' in the sense of rewarding directly specific actions (e.g., as in the setting of the Bitcoin blockchain where a miner who produces a block can obtain a certain amount of bitcoin as well as the fees of the transactions that are included in the block), 
or ``epoch based'' where the actions of all maintainers are examined in regular intervals and, based on performance, rewards are apportioned accordingly (e.g. in the case of the Cardano blockchain such epochs last 5 days). The distinction between action or epoch based is not very essential for the exposition of this section.

Let $\Omega$ be the finite universe of all resource units.  Resource units can be  exchanged between participants and some participants may hold a larger amount of resource units than others. What is a resource unit depends on the specific details of the system; it can be a hardware unit with software  that is capable of performing fast certain operations; it can be storage device; or, it can be a virtual unit controlled by a cryptographic key and maintained in a ledger. 

We consider that at system launch each resource unit is 
labelled by its current owner and, overloading notation, we will use $\Omega$
for the set of such labelled  units.  The {\em owner partition} of $\Omega$ is
a partition $O_1,\ldots,O_n$ that aggregates the units of all owners in
separate sets, where $n$ is the number of distinct owners. Such partitioning is dynamic, since resources can change hands and transfered between principals. 

Resource units owners need not engage with the system as maintainers directly. 
Instead they can form ``resource pools'' where many of them together operate a system maintenance node. In such configurations it is common that one of the contributors is the operator and the others invest in the system operation --- however other arrangements are also possible. Such pooling arrangements  can take the form of a contract between various resource holders enabling them to operate in tandem as an organization with members having different responsibilities. The Bitcoin system over the years has exhibited  significant pooling behavior and there were times that a single pool reached or even exceeded the critical threshold of controlling $51\%$ of the total active resources. 

We will use functions of the form  $c:2^\Omega \rightarrow \mathbb{R} \cup \{\bot\}$ to express the cost  that maps a set of units 
to the numerical cost expenditure that is incurred when the owners of these resource units engage in the system over a fixed period of time. 
Note that we only require $c$ to have a non $\bot$ value for 
sets of the owner partition of $\Omega$ and sets resulting by the
joining of these sets. 

A pooling configuration $\mc{P}$ is a family of mutually disjoint
sets $P_1,\ldots,P_m \subseteq\Omega$ accompanied by a reward splitting strategy 
for each pool that describes how to distribute rewards to the resource holders who participate (if they are more than one). It is important to note that 
rewards sharing at the protocol level only goes up to the level of the pool; 
beyond that, by the nature of resource based systems, it can be infeasible for the system to distinguish between a pool of a single resource holder compared to one where many join their resources.  

Pooling configurations are an important subject of study in resource based
systems since they reflect the ``decentralization'' of the underlying system. 
A stable centralized pooling configuration, e.g., one where all operators 
have joined a single pool, $\Omega$, indicates that the resource based system 
can be retired and substituted by a centralized system supported by an
organization reflecting the constituent membership of the single pool. 
In such circumstances, the benefits of using a resource based system entirely 
dissipate. As a result, it is of interest to understand in what settings such centralized pooling configurations may arise. 

Before we proceed, it is useful to introduce a metric for resource sets. 
We will use that metric to signify the influence that any single pool
can exert on the protocol. 
We denote the measure $\sigma : 2^\Omega \rightarrow \mathbb{R}$
of the resources of a pool $P$ by $\sigma(P)$. 
We require that $\sigma(\Omega) =1, \sigma(\emptyset)=0$ and 
$\sigma(P\cup Q) = \sigma(P) + \sigma(Q)  - \sigma(P\cap Q)$. 

Rewards sharing in resource based systems is controlled by a function
$\rho$; without loss of generality we count the rewards distributed in
a fixed period of time (the same period over which we also consider costs). 
Let $\rho(P,\mc{P})$ be the rewards
provided to a given pool $P$ by the system given a pooling configuration $\mc{P}$. 
A reward function 
$\rho(\cdot)$ is called {\em simple} if $\rho(P,\mc{P}) = \rho(P, \mc{P}') $
for any pooling configurations $\mc{P}, \mc{P}'$ that contain $P$. 
For simple reward functions we can write $\rho(P)$ to denote the rewards that
are provided to $P$. Note moreover
that $\rho(\emptyset) = 0$. 
A reward function is continuous if 
it holds that for every $P\subseteq \Omega, \epsilon>0$
there is a $\delta>0$ such that for any $P'$, 
$|\sigma(P)-\sigma(P')| < \delta \implies
|\rho(P)-\rho(P')| < \epsilon$. 
In the exposition of this section we consider only continuous simple reward 
functions. 

The reward function $\rho$ is a critical component of a resource based
system. We put forth the following set of axioms regarding the reward function 
$\rho $. As we will see, these axioms have certain implications regarding the pooling configurations that may arise in the system. 

\begin{itemize}
\item Resource fungibility. For any $P,Q$, 
$( \sigma(P) =\sigma(Q) )  \rightarrow ( \rho(P) = \rho(Q) )$. 
This means that the system does not distinguish between
particular resource units with respect to rewards. 



\item 
Sybil resilience. It holds that $\rho(P_1\cup P_2) \geq \rho(P_1) + \rho(P_2)$ for any disjoint sets $P_1,P_2\subseteq \Omega$.
This reflects the desideratum that an operator controlling some resources
will not gain anything by splitting their resources into two pools. 


\item 
Egalitarianism. It holds $\rho(P) \leq \rho(Q) + \rho(R)$ for any disjoint sets $P,Q,R$ such that $\sigma(P) = \sigma(Q) + \sigma(R)$. 
This reflects the desideratum that a ``rich''  operator controlling resources 
$P$ does not obtain more rewards than two ``poorer'' operators controlling 
in aggregate the same amount of resources. 

\end{itemize}

Given the above axioms, we will prove that a centralized pooling configuration can be a Nash equilibrium in the strong sense, i.e., even taking into account arbitrary coalitions~\cite{aumann}. We need two more properties to be defined first. 

\begin{definition}
Consider a pooling configuration $\mc{P}$.
A pool $P\in \mc{P}$ is called: 
(i) {\em viable},  if and only if $\rho(P)\geq c(P)$,
(ii) {\em cost efficient}, if and only if 
$c(P)/\sigma(P)\leq c(P')/\sigma(P')$, for any $P'\subseteq P$,
i.e., its cost per unit of resource is no worse than any of its subsets. 
\end{definition}

We are now ready to state and prove the following theorem. 

\begin{theorem}
If $\Omega$ is viable and cost efficient, then there is a centralized pooling configuration that is a Strong Nash equilibrium. 
\end{theorem}

\begin{proof}
Consider first any pooling configuration $\mc{P}$ 
and  $P\in \mc{P}$ such that it is viable and cost efficient. 
The rule to distribute rewards within $P$ is the following.  
Any   subset $S$ of $P$ corresponding to a participant 
receives rewards equal to $\sigma(S) (\rho(P)-c(P))$,
i.e., a ``fair'' share of the total rewards available. 
\ifFullElse{ 
We observe that: 
\begin{equation} \label{eq:theorem}  \sigma(S) (\rho(P) - c(P))
\geq \sigma(S) \rho(P) - \sigma(P) c(S) \geq \sigma(P)\rho(S)- \sigma(P)c(S)= \sigma(P) ( \rho(S) - c(S) )
\end{equation}
}{
We observe that: 
\begin{equation}
\begin{align*} \label{eq:theorem}  \sigma(S) (\rho(P) - c(P))
\geq \sigma(S) \rho(P) - \sigma(P) c(S) \geq \\ \sigma(P)\rho(S)- \sigma(P)c(S)= \sigma(P) ( \rho(S) - c(S) )
\end{align*}
\end{equation}
}

The first inequality follows  from 
the cost efficiency of $P$, which implies $c(S)/\sigma(S)\geq c(P)/\sigma(P)$. 

For the second inequality we need to prove $\rho(P)/\sigma(P) \geq \rho(S)/\sigma(S)$, i.e., the rewards per unit of resource is no worse
for $P$ compared to $S$. 
We will prove something stronger. 
For any $x\in [0,1]$ we define 
$\hat \rho(x)$  to be equal to the value $\rho(P)$ for some $P$ with $\sigma(P)=x$. 
The function $\hat \rho$ is well defined due to resource fungibility. 
Furthermore, observe that
$\hat \rho$ is superadditive due to Sybil resilience, 
and subadditive due to Egalitarianism. It follows that $\hat\rho$ 
satisfies Cauchy's functional equation and as a result,
due to the continuity of $\hat \rho(\cdot)$,
it holds that $\hat\rho(x) = \gamma x$,
for some $\gamma \in \mathbb{R}$. From this we derive that  $\rho(S) / \sigma(S) = \gamma = \rho(P)/\sigma(P)$. 
\ignore{
Consider the centralized pooling configuration $\{ \Omega \}$ 
with a pooling reward sharing agreement that assigns to any $P\subseteq \Omega$
the portion of available rewards  equal to $\sigma(P)$ after removing the operational costs. 

Suppose now there is a set of players controlling resources 
$P$ considering to break the pool and operate alone in a
pooling configuration $\mc{P} = \{ P , \Omega\setminus P\}$. 
The total payoff received by the diverging coalition in the alternative configuration
is $\rho(P) - c(P, \mc{P})$.

Observe that the rewards received by the players 
in the centralized pooling configuration are
$\sigma(P) (\rho(\Omega) - c(\Omega, \{\Omega\}))$. We prove the following. 

\[ \sigma(P) (\rho(\Omega) - c(\Omega, \{\Omega\}))
= \sigma(P) \rho(\Omega) - \sigma(P) c(\Omega, \{\Omega\}))
\geq \sigma(P) \rho(\Omega) - c(P,\mc{P}) \geq \rho(P)- c(P,\mc{P}) \]

Observe that the first inequality follows directly from 
economies of scale, since
$c(\Omega, \{\Omega\})) \leq \min \{ c(P,\{P, \Omega \setminus P\}) / \sigma(P), 
c(\Omega \setminus P ,\mc{P}) / \sigma(\Omega\setminus P) \}$. 

For the second inequality we need to prove $\rho(\Omega) \geq \rho(P)/\sigma(P)$. 
We consider the following. 
For any $x\in [0,1]$ we define 
$\hat \rho(x)$  to be equal to the value $\rho(P)$ for some $P$ with $\sigma(P)=x$. 
The function $\hat \rho$ is well defined due to resource fungibility. 
Furthermore, observe that $\hat \rho(1) = 1$, 
$\hat \rho$ is superadditive due to Sybil resilience, 
and subadditive due to Egalitarianism. It follows that $\hat\rho$ 
satisfies Cauchy's functional equation and as a result,
due to the continuity of $\hat \rho(\cdot)$,
it holds that $\hat\rho(x) = \gamma x$,
for some $\gamma \in \mathbb{R}$. We now have $\rho(P) / \sigma(P) \leq  \hat\rho(x)/x \leq \gamma \leq \hat\rho(1) \leq \rho(\Omega)$.

This completes the proof as the rewards received by $P$ within $\Omega$ are no worse
than operating a pool individually. 
}
 
We conclude by setting $P=\Omega$. Due to $\sigma(\Omega)=1$, by equation~\ref{eq:theorem}, the profit of  $S$, equal to $\rho(S)-c(S)$, is no better than the rewards received as part of the centralized configuration which equals to $\sigma(S)( \rho(\Omega) - c(\Omega))$. This implies that any set of participants will be no better off operating their own pool separating from the centralized pool $\Omega$. The same also holds in case they decide to run multiple separate pools. 
\end{proof}

It follows that it is of interest to detect large cost efficient resource sets.
To this end, we examine an important class of cost functions, that  we call
``operator-linear.''
First, let $O_1,\ldots,O_n$ be the owner partition of $\Omega$. 
The cost function is operator linear if it holds that  (i) for all $i=1,\ldots,n$, 
$c(O_i) = c_i + d_i \cdot \sigma(O_i) $, 
and (ii) for any $P = \cup^m_{j=1} O_{i_j}$, it holds the cost of $P$ is defined by  the following function
$$c(P) = d_{i_1}\cdot  \sigma(O_{i_1}) + \ldots d_{i_m} \cdot \sigma(O_{i_m}) + \min\{ c_{i_1}, \ldots, c_{i_m}\}.$$ 

This class of cost functions captures, at a certain level of abstraction, both proof of work and proof of stake systems where pooling is organized so that the operator becomes the resource holder with the smaller individual fixed cost. 
For proof of stake, given the cost incurred for processing is independent of the
resources held, one can set $d_i = 0$ for all $i=1,\ldots,n$. 
For proof of work, we observe the linear dependency in the amount of resources held that can be reflected by choosing a suitable value for $d_i$ derived from electricity costs and equipment characteristics used for performing the proof of work operation.
\ifFullElse{ 
We now prove the following proposition.}{ 
The following proposition is simple to prove; we leave the proof as an exercise.}

\begin{proposition}
\label{prop:costeff}
Given an operator-linear cost function, $\Omega$ is cost efficient, as long as $\Delta \leq \min\{c_1,\ldots,c_n\} $, where $\Delta =\max \{d_i - d_j \mid i,j \in [n]\}$. 
\end{proposition}

\ifFull{ 
\begin{proof}
We want to prove that $c(\Omega)\leq c(S)/\sigma(S)$ for any $S\subseteq \Omega$. 
Let us denote by $x_i = \sigma(O_i)$, where $O_1,\ldots,O_n$ is the owner partition of $\Omega$.  Without loss of generality we assume that 
$S$ includes the operators $O_1,\ldots,O_k$ for some $k\leq n$. 
We also denote by $c_j = \min\{c_1,\ldots,c_j\}$. We want to prove that
$$ ( c_n + \sum^n_{i=1} d_i x_i ) \cdot \sum^k_{i=1} x_k \leq c_k + \sum^k_{j=1} d_j x_j $$
We observe that based on the condition in the proposition's statement,
we have that $ \Delta \cdot \sum^k_{i=1} x_j \leq c_k$ which implies
that $\sum^k_{j=1} (d_i-d_j) x_j x_i \leq c_k x_i$. 
Summing for all $i=k+1,\ldots,n$, we have that 

$$\sum^n_{i=k+1} \sum^k_{j=1} (d_i-d_j) x_i x_j \leq c_k \sum^n_{i=k+1} x_i
\Rightarrow 
\sum^n_{i=k+1} \sum^k_{j=1} d_i x_i x_j \leq 
\sum^n_{i=k+1} \sum^k_{j=1} d_j x_i x_j  
+ c_k \sum^n_{i=k+1} x_i
$$

We add now in both sides of the inequality the terms $\sum^k_{i,j=1} d_ix_ix_j$ and
$c_k \sum^k_{i=1} x_i$ and by the observation $c_n\leq c_k$, we have the inequality 
$$ c_n \sum^k_{i=1} x_i + \sum^n_{i=1} \sum^k_{j=1} d_i x_i x_j 
\leq \sum^n_{i=1} \sum^k_{j=1} d_j x_i  x_j  + 
c_k \sum^n_{i=1} x_i
$$
From this we obtain $(c_n + \sum^n_{i=1} d_i x_i)\sum^k_{i=1} x_i   \leq c_k + \sum^k_{j=1} d_j x_j $ that proves $c(\Omega) \leq c(S) /\sigma(S)$.
\end{proof} 
} 

Based on the above, we obtain that  if $\Omega$ is viable
and the conditions of proposition~\ref{prop:costeff} are satisfied,
the system will have a strong Nash equilibrium that centralizes to one operator. 
This applies to both proof of stake as well as proof of work in the
case when differences in electricity costs are small across operators. 

On the other hand, in settings where cost efficiency does not hold, the joining of two resource
sets can become undesirable for one of the two operators. 
A weaker property for cost functions that captures 
``economies of scale''  and
dictates that $c(P_1\cup P_2) \leq c(P_1) + c(P_2)$, 
(reflecting the property that merging two pools results in no higher costs
compared to the two pools operating alone),
is insufficient by itself to imply a centralized pooling configuration. 

Even in the case of operator-linear cost functions however, 
careful design of the reward function and analysis of  Nash dynamics 
can show that better equilibria arise and can be reachable by the
participants. For instance, if costs for ``off-chain'' pooling are high,
the rewards sharing schemes developed and analyzed in \cite{DBLP:conf/eurosp/BrunjesKKS20} can be seen to converge 
to highly participatory decentralized equilibria for constant cost functions.

%% file: blueprint.tex
Given the four characteristics outlined in the previous sections, 
we will provide an illustration how to apply those to 
develop and deploy a {\em stake-based} system. 
We assume as preconditions that the developer has 
already a classical distributed protocol implementation 
of the service $\mc{F}$ for, say, $k$ parties 
and has an understanding of the service maintenance  costs and 
user demand. 

Adopting a stake-based approach, the resource will be digital coins.
The developer mints an initial supply of such coins and 
disperses them over an existing population of recipients. 
This can be achieved by e.g., ``airdropping'' such digital coins
to cryptocurrency holders of an existing blockchain platform. 
Due to this distribution event, the recipients become the stakeholders
of the system. 

A tokenomics schedule that takes into account the expected demand 
is determined and programmed into a smart
contract $\mc{S}$. This contract will acknowledge the initial supply of coins
as well as the schedule under which any new coins will be made available to the maintainers -- the entities running the $k$-party protocol. 
Following market-based tokenomics the contract will also manage
incoming transaction fees. 

Decentralized service provision is comprised of four parts.  
One is the $k$-party protocol that implements $\mc{F}$; 
the second is a proof-of-stake blockchain protocol that offers ``dynamic availability'' 
(e.g., such as Ouroboros, \cite{DBLP:conf/ccs/BadertscherGKRZ18,DBLP:conf/crypto/KiayiasRDO17}) -- i.e., a protocol that can handle 
a wide array of participation patterns without the requirement to be able
to predict closely the active participation level). Inputs to the protocol will be recorded on chain,
an action that will incur transaction costs to be withheld by $\mc{S}$. 
The third part is a ``proof-of-service'' 
sub-system that should enable any system
maintainer running the $k$-party protocol to demonstrate their efforts
in a robust way. The verifier of such proofs will be the smart contract 
$\mc{S}$ which will determine a performance factor for each maintainer. Finally, the fourth part is an algorithm that will
parse the blockchain at regular intervals and determine the $k$ parties to
run the $k$-party protocol for $\mc{F}$. This can be done e.g., by weighted sampling \cite{DBLP:reference/algo/EfraimidisS16}, taking into account the stake supporting each operator. 

For rewards sharing, we need a mechanism  to incentivize the stakeholders
to organize themselves into at least $k$ well functioning nodes 
that will execute the multiparty protocol for $\mc{F}$ when selected.
To achieve this we can deploy the reward sharing scheme of 
\cite{DBLP:conf/eurosp/BrunjesKKS20} over the underlying PoS blockchain; for that scheme it is  shown how incentive driven engagement by the stakeholders can determine a set of $k$ nodes at equilibrium. The reward scheme will be coded into the contract $\mc{S}$ and will reward the  stakeholders at regular intervals using the available supply from the tokenomics schedule and the transaction fees collected. The performance factor of each operator will influence the rewards, adjusting them in proportion to the operator's efforts. 

The developer will produce an implementation of the above system and will make it available for download. A launch date for the system will be set as well as an explanation for its purpose. At this point, the developer's engagement can stop. The stakeholders ---the recipients of the newly minted digital coin--- can examine the proposition that the system offers and choose whether to engage or not. If a non-negligible number of them chooses to engage out of their own self-interest (which will happen if the developer's predictions regarding the long term utility of $\mc{F}$ are correct) the system will come to life bootstrapping itself. 

%% file: conclusion.tex
In this paper, we put forth a new paradigm for deploying Information Technology
services inspired by the operation of the Bitcoin system. We 
identified four characteristics of resource based systems:
(i) resource-based operation, (ii) tokenomics, (iii) decentralized service provision, and (iv) rewards sharing, 
and
we elaborated on their objective and associated design challenges. 
We also presented a high-level blueprint showing how the paradigm 
can materialize in the form of a stake-based system. 

In more details, 
we identified the cryptographic 
and distributed protocol design challenge which asks  for a suitable PoX algorithm integrated with a protocol that facilitates decentralized service provision and the requirement of robust performance metrics.  
We also pointed to the economics and game theoretic considerations 
related to tokenomics and reward sharing. 

We   explored at some length
game theoretic aspects of pooling behavior 
and proved that a centralized equilibrium can be a strong Nash equilibrium
for a wide variety of reward and cost functions. This result is in the same spirit 
but more general than previous negative results presented in 
\cite{DBLP:conf/aft/KwonLKSK19,DBLP:conf/innovations/ArnostiW19,kiasto2021}
as it does not rely on the distribution of resources across owners, 
or a specific ``economies of scale''  assumption that dictates a superlinear
relation between rewards and costs, or the specific reward scheme used in Bitcoin, respectively. 

On a more positive note,  existence of  ``bad equilibria''   
does not prohibit the existence of other  equilibria with better
decentralization profiles. Furthermore, centralized pooling configurations
require coordination between agents that may prove difficult and costly to 
achieve. In this respect, ``bimodal'' systems, see \cite{kia2020}, where users can perform 
two (or more) actions to engage in system maintenance  (e.g., propose themselves as operators as well as vote others as operators) 
examples of which include \cite{DBLP:conf/eurosp/BrunjesKKS20,ceste2021}, show promise in this
direction. Furthermore, being able to investigate the Nash dynamics of the system  as e.g., performed in \cite{DBLP:conf/eurosp/BrunjesKKS20}  is crucial to 
demonstrate that the system reaches desirable equilibria expediently and moreover it can be also possible to demonstrate that bad equilibria can be avoided. 
It  is also worth pointing out that even if a resource based system manages to scale but eventually centralizes, the invested efforts may still not be completely in vain: the resulting constituent membership of the centralized pool organization may take over as a centralized system and offer the service.

The characteristics put forth in this paper are, in many respects, the minimum necessary. Other desirable features can be argued such as  the existence of multiple, open source software codebases that realize the system's protocol as well as  the existence of a governance sub-system that facilitates operations such as software updates not only for correcting the inevitable software bugs but also ensuring the system adapts to run time conditions that were unanticipated during the initial design. The problem of software updates in the decentralized setting is complex and more research is required, cf. \cite{DBLP:conf/esorics/CiampiKKZ20} for some first steps in terms of formally defining the problem in the context of distributed ledgers.  

The resource based paradigm is still in its very beginning. Nevertheless, 
we can identify some early precursors that include smart contract systems --- e.g.,  Ethereum and Cardano,  the name service of Namecoin, or the cross border payment system of Ripple. More recently, the Nym network \cite{nym} exemplified the paradigm in a novel context --- that of mix-nets and privacy-preserving communications. Extending the paradigm to additional use cases will motivate further advances in cryptography, distributed systems and game theory and eventually has the potential to change the landscape of global information technology.

%% file: acknowledgements.tex
I am grateful to Dimitris Karakostas,  Aikaterini Panagiota Stouka and Giorgos Panagiotakos for helpful comments and discussions. 